\documentclass[reqno,11pt]{amsart}
\usepackage{amsgen,amstext,amsbsy,amsopn}

\usepackage[T1]{fontenc}
\usepackage[utf8]{inputenc}

\usepackage{extarrows}
\usepackage{textcomp}
\usepackage{mathrsfs}

\usepackage{mathtools}
\usepackage{amsthm}
\usepackage{amssymb}
\usepackage{amsfonts}
\usepackage{amsmath}
\usepackage[ngerman, english]{babel}
\usepackage{microtype}
\usepackage{verbatim}
\usepackage{bigints}
\usepackage{tikz}
\usepackage{graphicx}

\usepackage{fancyhdr}

\usepackage{geometry}
\newtheorem{lemma}{Lemma}[section]
\newtheorem{theorem}{Theorem}
\newtheorem{proposition}[lemma]{Proposition}

\theoremstyle{definition}

\theoremstyle{definition}
\newtheorem{remark}[lemma]{Remark}
\theoremstyle{definition}

\newcommand\Tr{{\rm Tr\,}} 

\newcommand\nn\nonumber

\numberwithin{equation}{section}

\newcommand{\R}{\mathbb{R}}
 
\newcommand{\N}{\mathbb{N}}
\newcommand{\Z}{\mathbb{Z}} 
\newcommand{\C}{\mathbb{C}}
\newcommand{\D}{\mathcal{D}} 
 
\newcommand{\F}{\mathcal{F}}
 
\renewcommand{\H}{\mathcal{H}}

\newcommand{\grad}{\nabla}

\newcommand\Trs{{\rm Tr}_0\,}
\newcommand\Tro{{\rm Tr}_\Omega\,}

\newcommand{\td}{\mathrm{d}}

\newcommand{\be}{\begin{equation}}
\newcommand{\bea}{\begin{align}}
\newcommand{\eea}{\end{align}}

\DeclareMathOperator{\tr}{tr}

\begin{document}

\title[Persistence of translational symmetry in the BCS model]{Persistence of translational symmetry in the BCS model with radial pair interaction}

\author{A. Deuchert}

\address{Institute of Science and Technology Austria (IST Austria) \\ Am Campus 1, 3400 Klosterneuburg, Austria
 \\ Email: andreas.deuchert@ist.ac.at}

\author{A. Geisinger}

\address{Mathematisches Institut, Universit\"at T\"ubingen\\ Auf der
  Morgenstelle 10, 72076 T\"ubingen, Germany\\ Email: alissa.geisinger@uni-tuebingen.de}

\author{C. Hainzl}

\address{Mathematisches Institut, Universit\"at T\"ubingen\\ Auf der
  Morgenstelle 10, 72076 T\"ubingen, Germany\\ Email: christian.hainzl@uni-tuebingen.de}

\author{M. Loss}

\address{School of Mathematics, Georgia Tech\\ 686 Cherry Street Atlanta, GA 30332, USA \\ Email: loss@math.gatech.edu}

\begin{abstract}
We consider the two-dimensional BCS functional with a radial pair interaction. We show that the translational symmetry is not broken in a certain temperature interval below the critical temperature. In the case of vanishing angular momentum our results carry over to the three-dimensional case.
\end{abstract}

\maketitle

\section{Introduction}

In 1957 Bardeen, Cooper, and Schrieffer published their famous paper with the title "Theory of Superconductivity", which contained the first, generally accepted, microscopic theory of superconductivity. In recognition of this work they were awarded the Nobel prize in 1972. Originally introduced to describe the phase transition from the normal to the superconducting state in metals and alloys, BCS theory can also be applied to describe the phase transition to the superfluid state in cold fermionic gases. In this situation, one has to replace the usual non-local phonon-induced interaction in the gap equation by a local pair potential. Apart from being a paradigmatic model in solid state physics and in the field of cold quantum gases, the BCS theory of superconductivity, that is, the gap equation and the BCS functional show a rich mathematical structure, which has been well recognized. See \cite{Ode64,BilFan68,Van85,Yan91, McLYan00,Yan05} for works on the gap equation with interaction kernels suitable to describe the physics of conduction electrons in solids and \cite{HHSS08,FHNS07,HaiSei08c,HaiSei08a,BraHaiSei14,FreHaiSei12,FraLem16} for works that treat the translation-invariant BCS functional with a local pair interaction. The gap equation and the BCS functional are related in the way that the former is the Euler-Lagrange equation of the latter. One main question in the study of BCS theory is whether the gap equation
\begin{equation} \label{eq:gapDelta}
\Delta(p) = - \frac{1}{(2\pi)^{d/2}} \int_{\mathbb{R}^d} \hat{V}(p-q) \frac{\tanh\left( E(q)/2T \right)}{E(q)} \Delta(q) \, \text{d}q,
\end{equation}
with $E(q) = \smash{\sqrt{(q^2-\mu)^2 + \left| \Delta(q) \right|^{\vphantom{b}2}}}$ has a non-trivial solution, that is, one with $\Delta \neq 0$. If this is so, the system is said to be in a superconducting/superfluid state. The function $\Delta$ has the interpretation of a spectral gap of an effective mean-field Hamiltonian that is present only in the superconducting/superfluid phase, see the Appendix in \cite{HHSS08} for further explanations. In \cite{HHSS08} it has been demonstrated that, although the gap equation is highly non-linear, the question whether there exists a non-trivial solution can be decided with the help of a linear criterion. To be more precise, it was shown that the existence of a non-trivial solution of the gap equation is equivalent to the fact that a certain linear operator has a negative eigenvalue. Based on a characterization of the critical temperature in terms of this linear operator, its behavior has been investigated in the limit of small couplings and in the low-density limit, see \cite{FHNS07,HaiSei08a} and \cite{HaiSei08b}, respectively. Recently, there has also been considerable interest in the BCS functional with external fields, and in particular, in its connection to the Ginzburg-Landau theory of superconductivity, see \cite{HaiSei12,BraHaiSei16,HaiSch13,FHSS12,FHSS16,Deu17,FraHaiSchSei16,HaiSey16}.

The gap equation in the form stated in Eq.~\eqref{eq:gapDelta} and the related BCS functional can be heuristically derived from Quantum Mechanics by a variational procedure under several simplifying assumptions, see \cite{HHSS08} and the discussion in Section~\ref{sec:mainresults} below. One of these assumptions is that states used in this variational procedure are translation-invariant which leads to a strong simplification of the model. While this approximation is presumably valid in the case of cold fermionic gases with a rotationally-invariant pair interaction and is of great importance when it comes to numerical computations, it is in general hard to justify its validity. See \cite{Lieb94} for examples in the context of solid state physics where this approximation is not valid. From a mathematical point of view one is faced with a functional that is invariant under translations in the sense that spatial translations do not change the energy of a state. Due to the non-linear nature of the functional, minimizers need not be translation-invariant, however. If they are not one says that the translational symmetry of the system is broken. The aim of this work is to prove the absence of translational symmetry breaking in two situations: We start by considering the two-dimensional BCS functional with a radial pair interaction and show that there exists a certain temperature interval below the critical temperature, in which the translational symmetry of the system persists. Afterwards, we realize that our analysis directly carries over to the three-dimensional case if the Cooper-pairs are in an s-wave state. Prior to this work, such a result was known only in the case of $\hat{V} \leq 0$ and not identically zero, see \cite{HaiSei16}.

\section{Main Results}
\label{sec:mainresults}
We consider a sample of fermionic atoms in a cold gas in $d$-dimensional space ($d=2,3$) within the framework of BCS theory. It is convenient to think of the sample as infinite and periodic, since this setting avoids having to deal with boundary conditions at the boundary of the sample. To describe the periodicity we introduce the lattice $\Z^d$ with the unit cell $[0,1]^d = \Omega$. The special form of the lattice does not play any role for us and the proof carries over to an arbitrary Bravais lattice. To not artificially complicate the presentation, we therefore opt for the simplest choice. BCS states are most conveniently described by their generalized one-particle density matrix, that is, by a self-adjoint operator~$\Gamma$ on~$L^2(\R^d) \oplus L^2(\R^d)$ of the form
\begin{align} \label{eq:formGammaper}
\Gamma = \left( \begin{array}{cc} \gamma & \alpha \\ \overline{\alpha} & 1 - \overline{\gamma}  \end{array} \right), 
\end{align}
with $0 \leq \Gamma \leq 1$. Here $\gamma$ and $\alpha$ denote the one-particle density matrix and the Cooper-pair wave function of the state $\Gamma$, respectively. Both of them are represented by periodic operators with period one. In terms of kernels, the latter means that $\gamma(x + u, y + u) = \gamma(x,y)$ and $\alpha(x + u, y + u) = \alpha(x,y)$ for all $u \in \Z^d$ and all $x,y \in \R^d$. In~\eqref{eq:formGammaper}, $\overline{\alpha} = C\alpha C$, where~$C$ denotes complex conjugation. Note that, in particular, $\alpha(x,y) = \alpha(y,x)$ for all $x,y \in \R^d$, due to the self-adjointness of~$\Gamma$. In this setting, it is natural to consider energies per unit volume. Accordingly, we define for a periodic operator~$A$, the trace per unit volume~$\Tro$ by $\Tro[A] = \Tr[\chi_{\Omega}A\chi_{\Omega}]$, where~$\chi_{\Omega}$ denotes the characteristic function of~$\Omega$. We call~$\Gamma$ of the form~\eqref{eq:formGammaper} an \textit{admissible} BCS state if 
$\Tro(- \grad^2 + 1)\gamma < \infty$ and denote the set of admissible BCS states by~$\D$. We will, by a slight abuse of notation, write $(\gamma, \alpha) \in \D$, meaning that the BCS state~$\Gamma$ given by \eqref{eq:formGammaper} is admissible. 

The BCS functional at temperature~$T \geq 0$, with chemical potential~$\mu \in \R$, interaction potential~$V \in L^{2}(\R^d)$ and entropy  
\begin{align*}
S(\Gamma) = - \frac{1}{2}  \Tro [\Gamma\log\Gamma+ \left(1 - \Gamma\right) \log\left(1 - \Gamma\right)],
\end{align*}
is then given by
\begin{align} \label{def:functionalPER}
\F(\Gamma) = \Tro \left[\left(- \grad^2 - \mu\right)\gamma\right] + \int_{\Omega \times \R^d} V(x - y) \vert \alpha(x,y) \vert^2 \, \td(x,y) - TS(\Gamma).
\end{align}
Note that the same functional has been considered in~\cite{FHSS12}, where the periodicity was introduced for ease of comparison with the translation-invariant functional. As already mentioned above, the BCS functional can be heuristically derived from Quantum Mechanics by a variational procedure. To that end, one considers the full free energy functional of the system and restricts attention to quasi-free states only. Due to the Wick rule, the energy and the entropy can then be expressed solely in terms of the generalized one-particle density matrix of the quasi-free state under consideration, see \cite{Lieb94}. If one assumes additionally $SU(2)$-invariance as well as the above periodicity of the state and neglects the direct and the exchange term in the energy, one arrives at Eq.~\eqref{def:functionalPER}. For more details see the Appendix of \cite{HHSS08}.

The translation-invariant BCS functional~$\F^{\mathrm{ti}}$ is obtained from~$\F$ by restricting the set of admissible states to the translation-invariant ones. That is, the kernels of~$\gamma$ and~$\alpha$ take the form $\gamma(x,y) = \gamma(x-y)$ and $\alpha(x,y) = \alpha(x-y)$, respectively. 
We describe translation-invariant BCS states via their momentum representations by $2\times2$ matrices of the form
\begin{align} \label{eq:formGamma}
\hat{\Gamma}(p) = \left( \begin{array}{cc} \hat{\gamma}(p) & \hat{\alpha}(p) \\ \overline{\hat{\alpha}(p)} & 1 - \hat{\gamma}(-p)  \end{array} \right),
\end{align}
for $p \in \R^d$, where the bar denotes complex conjugation and the hats indicate that those objects are Fourier transforms of integral kernels that depend only on $x-y$. Obviously, $\hat{\Gamma}(p)$ satisfies $0 \leq \hat{\Gamma}(p) \leq 1$ for all $ p \in \R^d$. The latter translates to $\vert \hat{\alpha}(p) \vert^2 \leq \hat{\gamma}(p)(1 - \hat{\gamma}(p))$ for $p \in \R^d$ in terms of~$\hat{\gamma}$ and~$\hat{\alpha}$. Note that the fact that $\Gamma$ is self-adjoint implies that $\hat{\alpha}$ is an even function and that $\hat{\gamma}$ is real-valued. A translation-invariant BCS state~$\Gamma$ is admissible if and only if $\hat{\gamma} \in L^1(\R^d, (1+p^2)\, \td p)$ and $\alpha \in H^1(\R^d, \td x)$. By~$\D^{\mathrm{ti}}$ we denote the set of all admissible translation-invariant BCS states. For~$T \geq 0$ the translation-invariant BCS functional with chemical potential~$\mu \in \R$, interaction potential~$V \in L^{2}(\R^d)$ and entropy~$S$, which we can now write as
\begin{align*}
S(\Gamma) = - \frac{1}{2} \int_{\R^d} \tr_{\C^2} \left[\hat{\Gamma}(p)\log\hat{\Gamma}(p) + \left(1 - \hat{\Gamma}(p)\right) \log\left(1 - \hat{\Gamma}(p)\right) \right] \, \td p,
\end{align*}
takes the form
\begin{align} \label{eq:functionalTI}
\F^{\mathrm{ti}}(\Gamma) = \int_{\R^d} (p^2 - \mu)\hat{\gamma}(p) \, \td p + \int_{\R^d} V(x) \vert \alpha(x)\vert^2 \, \td x - TS(\Gamma).
\end{align}
Given a state $\Gamma$, we define the gap function $\Delta$ of that state as the Fourier transform of $2 V(x) \alpha(x)$. One can then show that the gap function of any minimizing BCS state satisfies Eq.~\eqref{eq:gapDelta}, see \cite{HHSS08}. We note that $\F^{\mathrm{ti}}$ was studied in \cite{HHSS08} without the constraint that $\alpha$ is reflection symmetric. The results there hold equally if one works only in the subspace of reflection-symmetric functions in $L^2(\R^d)$, however. In the case of $V = 0$, the translation-invariant BCS functional $\F^{ti}$ is minimized by the pair $(\gamma_0,0)$ where $\hat{\gamma}_0(p) = (1+e^{\beta(p^2-\mu)})^{-1}$. The same statement is true for the periodic BCS functional $\F$. The state $(\gamma_0,0)$ is called the normal state and describes a situation where superfluidity is absent.

It was shown in~\cite[Theorem 1]{HHSS08} that there exists a critical temperature $T_c \geq 0$ such for $T < T_c$, the minimizer of the translation-invariant BCS functional has a non-vanishing Cooper-pair wave function. On the other hand, for $T \geq T_c$, the normal state is the unique minimizer. Additionally, there is a characterization of $T_c$ in terms of a linear operator. To make this statement more explicit, let us introduce the function $K_T : \R^d \to \R$ given by
\begin{align*}
K_T(p) = \frac{p^2 - \mu}{\tanh((p^2 - \mu)/(2T))}.
\end{align*}
Then, $K_T = K_T(-i\grad)$ defines an operator on $L^2(\R^d)$ acting by multiplication with~$K_T(p)$ in Fourier space. The critical temperature of the translation-invariant BCS functional is given by
\begin{align*}
T_c = \inf\lbrace T \geq 0 \ | \ K_T + V \geq 0 \rbrace.
\end{align*}
In other words, $T_c$ is the value of $T$ such that the operator $K_T + V$ has zero as lowest eigenvalue. Observe that this definition makes sense because $K_T$ is monotone increasing in~$T$. The characterization of $T_c$ in terms of a linear operator comes about because a minimizer of the translation-invariant BCS functional $\F^{ti}$ has a non-vanishing Cooper-pair wave function if and only if the normal state is unstable under pair formation. That is, if and only if the second variation of $\F^{ti}$ at $(\gamma_0,0)$ has a negative eigenvalue. The operator $K_T + V$ is exactly the second variation of $\F^{ti}$ at the normal state in the direction of a perturbation with $\gamma = 0$ and $\alpha \not\equiv 0$.  

In this paper, we treat the question whether there is translational symmetry breaking in the BCS model with radial pair interaction $V$. More precisely, we study the minimization problem
\begin{align*}
\inf\left\{ \F(\Gamma) \, \vert \, \Gamma \in \D \right\}
\end{align*}
and we are, in particular, concerned with the question whether the infimum of~$\F$ is attained by the minimizers of the translation-invariant BCS functional. If $\hat{V} \leq 0 $ with $\hat{V}$ not identically zero this is already known to be the case, see \cite{FHSS12,HaiSei16}. In order to study this question, we consider the BCS functional~$\F_{\ell}^{\mathrm{ti}}$ on the sector of translation-invariant BCS states with Cooper-pair wave functions of angular momentum $\ell \in 2\N_0$, that we will define in the next paragraph. Our strategy consists of showing that there exists~$\ell_0$  such that the minimizers of $\F_{\ell_0}^{\mathrm{ti}}$ and~$\F$ coincide under certain assumptions. 

Let us now introduce the functionals~$\F_{\ell}^{\mathrm{ti}}$ in the case $d=2$. They are obtained from $\mathcal{F}^{\mathrm{ti}}$ by restricting the domain to Cooper-pair wave functions of the form \begin{align}\label{eq:formofalpha}
\hat{\alpha}_{\ell}(p) = e^{i\ell\varphi} \sigma_{\ell}(p),
\end{align}
for some $\ell \in 2\Z$, where~$\varphi$ denotes the angle of $p \in \R^2$ in polar coordinates  and~$\sigma_{\ell}$ is a radial function. Recall that~$\alpha$ is an even function, which requires~$\ell$ to be even. As we will see, the Euler-Lagrange equation of $\F^{\mathrm{ti}}$ implies that if $(\gamma, \alpha_{\ell})$ is a minimizer of~$\F^{\mathrm{ti}}$, then~$\hat{\gamma}$ has to be a radial function. Therefore, we define
the BCS functional on the sector of Cooper-pair wave functions of angular momentum~$\ell$ as follows. We make an angular decomposition for $(p,q) \mapsto \hat{V}(p - q)$, that is
\begin{align*}
\hat{V}(p-q) = \sum_{\ell \in \Z} \hat{V}_{\ell}(p,q) e^{i\ell\varphi},
\end{align*}
where~$\varphi$ denotes the angle between~$p$ and~$q$. In other words, this means that
\begin{align}\label{eq:Vell} 
\hat{V}_{\ell}(p,q) = \frac{1}{2\pi} \int_0^{2\pi} e^{-i\ell\varphi} \hat{V}(p-q) \, \td \varphi.
\end{align}
Since~$\hat{V}$ is a radial function, it only depends on the absolute value of its argument, that is, on $\vert p - q\vert = \smash{\sqrt{p^2 + q^2 - 2\vert p \vert \vert q \vert \cos(\varphi)}}$ and we conclude that~$\hat{V}_{\ell}$ is radial in both arguments. Furthermore, observe that $\hat{V}_{\ell} = \hat{V}_{-\ell}$.

Then, the BCS functional~$\F_{\ell}^{\mathrm{ti}}$ on the sector of Cooper-pair wave functions of even angular momentum $\ell \in 2\N_0$ is given by
\begin{align*}
\F_{\ell}^{\mathrm{ti}}(\Gamma_{\ell}) = \int_{\R^2} (p^2 - \mu) & \gamma_{\ell}(p) \, \td p + \int_{\R^2} \int_{\R^2} \overline{\sigma_{\ell}(p)} \sigma_{\ell}(q)  \hat{V}_{\ell}( p, q) \,\td p \td  q - TS(\Gamma_{\ell}),
\end{align*}
where~$V_{\ell}$ is given in~\eqref{eq:Vell} and~$\Gamma_{\ell}$ is determined by the pair $(\gamma_{\ell}, \sigma_{\ell})$ with radial functions~$\gamma_{\ell}$ and~$\sigma_{\ell}$. To be more precise, the domain of~$\F_{\ell}^{\mathrm{ti}}$ is given by
\begin{align*}
\mathcal{D}_{\ell} \coloneqq \big\{ (\gamma_{\ell}, \sigma_{\ell}) \vert \, \gamma_{\ell}, \sigma_{\ell} \text{ radial and } (\gamma_{\ell}, \alpha_{\ell}) \in \D^{\mathrm{ti}}, \, \hat{\alpha}_{\ell}(p) = e^{i\ell\varphi}\sigma_{\ell}(p) \text{ for } p \in \R^2 \big\}.
\end{align*}
Equivalently, $\F_{\ell}^{\mathrm{ti}}$ can be understood as the restriction of~$\F^{\mathrm{ti}}$ to pairs $(\gamma, \alpha) \in \D^{\mathrm{ti}}$ with the property that~$\gamma$ is radial and that~$\alpha$ is of the form given in~\eqref{eq:formofalpha}. In Section~\ref{sec:prep} we will show that~$\F_{\ell}^{\mathrm{ti}}$ has a minimizer. 

Next, we characterize the critical temperature~$T_c(\ell)$ corresponding to the BCS functionals~$\F_{\ell}^{\mathrm{ti}}$ on the sector of Cooper-pair wave functions of angular momentum $\ell \in 2\N_0$. For this purpose, let us introduce $\mathcal{H} = \{ f \in H^1(\R^2, \td p) \ \vert \ f \text{ radial }\}$. Then the critical temperature~$T_c(\ell)$ of~$\F_{\ell}^{\mathrm{ti}}$ is given by
\begin{align} \label{eq:TcMax}
T_c(\ell) \coloneqq \inf \left\{ T \geq 0 \left\vert \, \left(K_T + V_{\ell}\right)\big\vert_{\H} \geq 0 \right. \right\}.
\end{align}
The definition of $V_{\ell}$ in Eq.~\eqref{eq:Vell} and the fact that $K_T + V$ commutes with rotations, implies that
\begin{align*}
T_c = \max_{\ell \in 2\N_0} T_c(\ell)
\end{align*}
holds.

Let us now assume that $T_c = T_c(\ell_0)$ and that the lowest eigenvalue of $K_{T_c} + V$ is at most twice degenerate. In other words, we assume the lowest eigenvalue of $K_{T_c} + V$ to be exactly twice degenerate in the case $\ell_0 \neq 0$ and we assume it to be non-degenerate in the case $\ell_0 = 0$. An exemplary situation satisfying this assumption is illustrated in Figure 1. The meaning of this schematic pictures is the following. Since $T_c = T_c(\ell_0)$, the lowest eigenvalue of $K_T+V$ lies in the sector with angular momentum $\ell_0$. If we decrease the temperature this eigenvalue becomes negative and the second/third  eigenvalue (depending on the degeneracy) will approach zero at some temperature $\tilde{T} < T_c(\ell_0)$. For this eigenvalue, there are two possibilities: Either it also lies in the sector of angular momentum $\ell_0$, which means that $\tilde{T} \in (T_c(\ell_1),T_c)$ and this is the case illustrated in Figure 1, or the next eigenvalue lies in the next sector of angular momentum, which means that $\tilde{T} = T_c(\ell_1)$. 
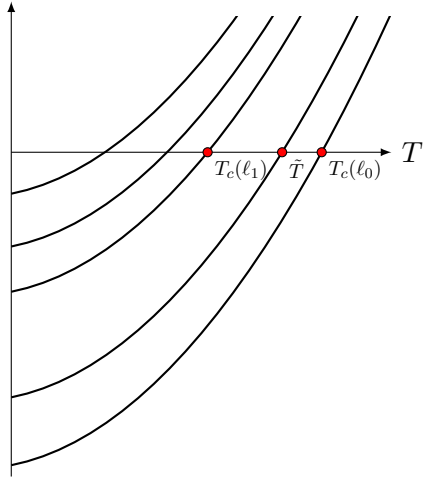
\begin{figure}
	\centering
	\begin{tikzpicture}[scale=1, >=latex]

\def\ymin{-4.3}
\def\ymax{2}
\def\ytcn{-4.2}
\def\ytcd{-3.3}
\def\ytce{-1.9}

\draw[->] (0,0) -- (5,0) node[right]{$T$};
\draw[->] (0,\ymin) -- (0,\ymax);

\clip (0,\ymin) rectangle (5,{0.9*\ymax});
\foreach \c in {-.6,-1.3,\ytce, \ytcd, \ytcn}{%
    \draw[thick] plot[domain=0:5](\x,{.2*(\x+.5)^2+\c});
}

\filldraw[black, fill=red] ({sqrt(-5*\ytcn)-.5},0) circle(.06) node[below right, scale=.7]{$T_c(\ell_0)$};
\filldraw[black, fill=red] ({sqrt(-5*\ytce)-.5},0) circle(.06) node[below right, scale=.7]{$T_c(\ell_1)$};
\filldraw[black, fill=red] ({sqrt(-5*\ytcd)-.5},0) circle(.06) node[below right, scale=.7]{$\tilde{T}$};

\end{tikzpicture}
	\caption{Schematic picture of the lowest eigenvalues of $K_{T} + V$ as a function of the temperature $T$. The lowest two lines represent eigenvalues in the sector of angular momentum $\ell_0$. The third line corresponds to the lowest eigenvalue in the angular momentum $\ell_1$ sector. The red dots highlight the temperatures at which one of the eigenvalues crosses the $T$-axis.}
\end{figure}
The following theorem shows that the translational symmetry in the BCS model persists if $T \in (\tilde{T}, T_c)$. In particular, if $\ell_0 = 0$, the periodic (and the translation-invariant) BCS functional has a, up to a phase, unique radial minimizer $(\gamma_0, \alpha_0)$ for $T \in (\tilde{T}, T_c)$. If $\ell_0 \neq 0$, the periodic (and the translation-invariant) BCS functional has two minimizers, namely $(\gamma_{\ell_0}, \alpha_{\ell_0})$ and $(\gamma_{\ell_0}, \alpha_{-\ell_0})$, with $\gamma_{\ell_0}$ radial and $\alpha_{\pm\ell_0}$ of the form $\hat{\alpha}_{\pm \ell_0}(p) = e^{\pm i \ell_0 \varphi} \sigma_{\ell_0}(p)$.

\begin{theorem} \label{thm:vollesFunktional}
Let $V \in L^2(\R^2)$ with $\hat{V} \in L^{r}(\R^2)$, where $r \in [1,2)$, be radial and such that $T_c > 0$. Suppose that $T_c = T_c(\ell_0)$ and that the lowest eigenvalue of $K_{T_c} + V$ is at most twice degenerate.
If
\begin{align*}
(\gamma_{\ell_0}, \sigma_{\ell_0}) \in \mathcal{D}_{\ell_0} 
\end{align*}
minimizes~$\F_{\ell_0}^{\mathrm{ti}}$, then there exists $\tilde{T} < T_c$ such that
\begin{align*}
(\gamma_{\ell_0}, \alpha_{\ell_0}) \text{ and } (\gamma_{\ell_0}, \alpha_{-\ell_0}) \in \D^{\mathrm{ti}},
\end{align*}
where $\hat{\alpha}_{\pm \ell_0}(p) = e^{\pm i \ell_0 \varphi} \sigma_{\ell_0}(p)$, minimize the BCS functional $\F$ for $T\in [\tilde{T}, T_c)$. For $T \in (\tilde{T}, T_c)$ these are the only minimizers of~$\F$ up to phases in front of~$\alpha_{\ell_0}$ and~$\alpha_{-\ell_0}$.
\end{theorem} 

\begin{remark} \label{rem:ell1}
We want to emphasize that $\tilde{T}$ is determined by the lowest nonzero eigenvalue of $K_{T_c} + V$. More precisely, $\tilde{T}$ is given as the value of $T$ such that the second eigenvalue (counted without multiplicities) of $K_{T} + V$ is zero, which is illustrated in Figure 1.
In particular, if in addition to the assumption above, the second eigenvalue of $K_{T_c} + V$ lies in the sector of angular momentum $\ell_1 \neq \ell_0$, one can show that $\tilde{T} = T_c(\ell_1)$.
\end{remark}

\begin{remark}
The assumptions $V \in L^2(\mathbb{R}^2)$ and $\hat{V} \in L^{r}(\R^2)$ with $r \in [1,2)$ in Theorem~\ref{thm:vollesFunktional} are of technical nature and we expect the Theorem to hold as long as $V \in L^{1+\epsilon}(\mathbb{R}^2)$ for $\epsilon>0$. Note that this is the $L^p$ regularity for which $V$ is relatively form bounded with respect to the Laplacian in two space dimensions. The assumption on the Fourier transform of $V$ is only needed in the proof of Proposition~\ref{prop:nonnegativity}. In \cite[Proposition 5.6]{FraLem16} a similar result is proved in the case $d=3$ under the assumption $V \in L^{3/2}(\mathbb{R}^3)$ which guarantees form boundedness relative to the Laplacian in this case. Although we expect the strategy of that proof to carry over to $d=2$, our argument is much simpler than the one given in this reference and so we prefer to keep the additional assumption on $\hat{V}$.
\end{remark}
\begin{remark}
The Fourier transform preserves angular momentum sectors, and hence the inverse Fourier transforms of the minimizing Cooper-pair wave functions $\hat{\alpha}_{\pm \ell_0}(p) = e^{\pm i \ell_0 \varphi_p} \sigma_{\ell_0}(p)$ are of the form $e^{\pm i \ell_0 \varphi_x} f_{\ell_0}(x)$ with $f_{\ell_0}$ radial. That is, the Cooper-pairs have definite angular momentum also in position space.
\end{remark}
 
\begin{remark}
An important step in the proof of Theorem \ref{thm:vollesFunktional} is to compare the minimizers of the BCS functional~$\F_{\ell_0}^{\mathrm{ti}}$ on the sector of Cooper-pair wave functions with angular momentum $\ell_0$ with the minimizers of the periodic BCS functional~$\F$. The crucial tool for this comparison will be the relative entropy inequality,~\cite[Lemma 5]{FHSS12}.
\end{remark}

\begin{remark}
It is shown in~\cite{FraLem16}, amongst other things, that for every~$\ell \in 2\N_0$ one can find a radial potential such that the ground state of  $K_{T_c}+V$ has angular momentum~$\ell$. This in particular implies $T_c = T_c(\ell)$ for such a potential.
In the case of weak coupling, that is for $K_T + \lambda V$, where $\lambda \in \R$ is small enough, the methods of~\cite{FHNS07,HaiSei08a} can be applied to determine the angular momentum~$\ell_0$ of the ground state of~$K_{T_c}+V$. An application of these methods reduces the problem of finding the eigenvalues of $K_T + \lambda V$, for $\lambda$ small enough, to finding the eigenvalues of a simple matrix, that only depends on the behavior of $\hat{V}$ on the Fermi sphere. This is easily solvable numerically. In particular, one sees, that  the eigenvalues are in one-to-one correspondence to the eigenvalues of the matrix $\smash{( \langle \psi_n, \hat{V} \psi_m \rangle )_{n,m \geq 0}}$, where $\psi_n(p) = e^{i n \varphi}$. Moreover, if the lowest eigenvalue of this matrix is at most twice degenerate one is in the situation described in Remark \ref{rem:ell1}, i.e. $\tilde{T} = T_c(\ell_1)$.
\end{remark}

\begin{remark}
In the non-interacting case, that is, for~$V = 0$, the minimizer of the BCS function~$\F$ is given by the normal state
\begin{align*}
\hat{\Gamma}_0 = \left( \begin{array}{cc} \hat{\gamma}_0 & 0 \\ 0 & 1 - \overline{\hat{\gamma}_0} \end{array} \right),
\end{align*}
where $\hat{\gamma}_0 = (1 + \exp((-\grad^2 - \mu)/T))^{-1}$. 
Let us assume that we are in the situation of Remark \ref{rem:ell1}.
Having in mind that the linear operator~$K_T + V$, which characterizes~$T_c$, is related to the second variation of~$\F$ at the normal state~$\Gamma_{0}$ in the direction of~$\alpha$ by
\begin{align*}
\left. \frac{\td^2}{\td t^2} \F(\gamma_0, t \alpha)\right\vert_{t = 0} = 2\langle \alpha, (K_T + V)\alpha \rangle,
\end{align*}
one can understand Theorem~\ref{thm:vollesFunktional} as follows. We find $T < T_c$ such that~$K_{T} + V$ has exactly one negative eigenvalue~$\lambda_0$. Hence the second variation is smallest (and, in particular, negative) if $\alpha$ is an element of the eigenspace of~$\lambda_0$ and one could therefore hope to find a minimizer of~$\F$ which lies approximately in this eigenspace. In fact, Theorem \ref{thm:vollesFunktional} states that the minimizers of~$\F$ for temperatures~$T$ in a certain interval below~$T_c$ lie in exactly one specific sector of angular momentum~$\pm\ell_0$. For~$T = T_c(\ell_1)$ the next eigenvalue~$\lambda_1$ and its eigenspace become important, since now also elements of the eigenspace of~$\lambda_1$ are candidates to lower the energy. 
\end{remark}

In the special case~$\ell_0 = 0$, Theorem~\ref{thm:vollesFunktional} also holds in three dimensions. 

\begin{theorem} \label{thm:3dim}
Let $V \in L^2(\R^3)$ with $\hat{V} \in L^{r}(\R^3)$ for some $r \in [1,12/7)$ be radial and such that $T_c > 0$. Assume that zero is a non-degenerate eigenvalue of $K_{T_c} + V$, that is, the corresponding eigenfunction is radial.
Then, there exists $\tilde{T} < T_c$ such that the minimizer of the BCS functional $\F$ for $T \in [\tilde{T}, T_c)$ is given by a pair $(\gamma_0, \alpha_0)$, where $\gamma_0$ and $\alpha_0$ are radial functions. Moreover, $(\gamma_0, \alpha_0)$ is, up to phases, the only minimizer of $\F$ for $T \in (\tilde{T},T_c)$.
\end{theorem}

\begin{remark}
	Note that $\hat{V} \leq 0$ implies that the ground state of $K_{T_c} + V$ is radial in all dimensions. Hence, the assumption that $K_{T_c} + V$ has a non-degenerate lowest eigenvalue is always satisfied for interaction potentials $V$ with this property.
\end{remark}

\begin{remark}
As in the case of Theorem~\ref{thm:vollesFunktional}, we expect Theorem~\ref{thm:3dim} to hold under the only assumption that $V$ is relatively form bounded with respect to the Laplacian, that is, if $V \in L^{3/2}(\mathbb{R}^3)$.
\end{remark}
 
We recall the gap function $\Delta(p) = 2(2\pi)^{-d/2} \hat{V}\ast\hat{\alpha}(p)$ with $d=2,3$. The Cooper-pair wave function of any minimizer of the translation-invariant BCS functional $\mathcal{F}^{\mathrm{ti}}$ satisfies the Euler-Lagrange equation 
\begin{equation}
\left( K_T^{\Delta} + V \right) \alpha = 0.	
\label{eq:proposition1}
\end{equation}
Here $K_T^{\Delta}$ is the operator defined by multiplication in Fourier space with the function
\begin{equation}
K_T^{\Delta}(p) = \frac{E(p)}{\tanh\left(E(p)/(2T)\right)}, \quad \text{ where } \quad E(p) = \sqrt{(p^2-\mu)^2 + \vert \Delta(p) \vert^2}. \nonumber
\end{equation}
The key ingredient to the proof of Theorem~\ref{thm:vollesFunktional} and Theorem~\ref{thm:3dim} is that in both situations $K_T^{\Delta} + V \geq 0$ holds. The following Proposition tells us that this already implies that $\vert \hat{\alpha}(p) \vert$ is a radial function. Hence, our strategy of proof can only work if this is the case. In particular, it tells us that we cannot extend our results to situations where the absolute value of the Fourier transform of the ground state of $K_{T_c} + V$ is not radial.
\begin{proposition}
\label{prop:mainproposition}
Let $V$ be a radial function with $V \in L^{2}(\mathbb{R}^2)$ if $d=2$ and $V \in L^{3/2}(\mathbb{R}^3)$ if $d=3$. Assume that $(\gamma,\alpha)$ is a minimizer of the translation-invariant BCS functional $\mathcal{F}^{\mathrm{ti}}$ such that $\vert \hat{\alpha}(p) \vert$ is not a radial function. Then there exists a rotation $R \in SO(d)$ such that
\begin{equation}
\left\langle U(R) \alpha , \left( K_T^{\Delta} + V \right) U(R) \alpha \right\rangle < 0,
\end{equation}
where $\left( U(R) f \right)(p) = f(R^{-1}p)$.
\end{proposition}
\section{Preparations} \label{sec:prep}
The proof of Theorem \ref{thm:3dim} works similarly to the proof of Theorem \ref{thm:vollesFunktional}. In order to prove Theorem \ref{thm:vollesFunktional} we will show that there exists $\ell_0 \in 2\N_0$, such that the minimizers of~$\F_{\ell_0}^{\mathrm{ti}}$ also minimize~$\F$. The following lemma lays the basis for this approach.  

In~\cite{HHSS08} it was shown that~$\F^{\mathrm{ti}}$ is bounded from below and attains its infimum on~$\D^{\mathrm{ti}}$ in three dimensions. The same results hold in two dimensions by analogous arguments, which provides a solution of the BCS gap equation in this case. 

\begin{lemma} \label{lem:extmin}
The BCS functional~$\F_{\ell}^{\mathrm{ti}}$ is bounded from below and attains its minimum. 
\end{lemma}

\begin{proof}
Boundedness from below of~$\F_{\ell}^{\mathrm{ti}}$ follows from the fact that~$\F^{\mathrm{ti}}$ is bounded from below.
As in the proof of~\cite[Lemma 1]{HHSS08} we find a minimizing sequence $\smash{(\gamma_{\ell}^{(n)\vphantom{b}}, \sigma_{\ell}^{(n)\vphantom{b}})}$ in~$\D_{\ell}$ that converges strongly in $L^p(\R^2) \times L^2(\R^2)$  to $(\gamma,\sigma)$ for some $p \in (1, \infty)$, as~$n$ tends to infinity. It is an easy consequence that $(\gamma,\sigma) \in \mathcal{D}_{\ell}$.
\end{proof}


The Euler-Lagrange equation of~$\F_{\ell}^{\mathrm{ti}}$ takes the same form as the Euler-Lagrange equation of~$\F^{\mathrm{ti}}$, which will play an important role in the proof. The derivation of the Euler-Lagrange equation of~$\F^{\mathrm{ti}}$ given in~\cite[Proposition 3.1]{HaiSei16} translates to the case of~$\F_{\ell}^{\mathrm{ti}}$. Therefore, we will not rewrite the proof here, but only give the Euler-Lagrange equation of~$\F_{\ell}^{\mathrm{ti}}$ in its various forms. 

Let us define the gap function $\Delta_{\ell}$ related to the Cooper-pair wave function $\sigma_{\ell}$ by
\begin{equation}
\Delta_{\ell}(p) = \frac{1}{\pi} \int_{\mathbb{R}^2} V_{\ell}(p,q) \sigma_{\ell}(p) \td q.
\label{eq:gapfunction}
\end{equation}
Since $V_{\ell}(p,q)$ is radial in both arguments $\Delta_{\ell}(p)$ is a radial function, too. Also define 
\begin{equation} \label{def:HDelta}
	H_{\Delta_{\ell}}(p) = \left( \begin{array}{cc} k(p) & \Delta_{\ell}(p) \\ \overline{\Delta_{\ell}(p)} &  -k(p)  \end{array} \right)
\end{equation}  
with $k(p) = p^2 - \mu$. For $T > 0$, the Euler-Lagrange equation of the functional~$\F_{\ell}^{\mathrm{ti}}$, is given by
\begin{equation}
	\Gamma_{\ell}(p) = \begin{pmatrix} \gamma_{\ell}(p) & \sigma_{\ell}(p) \\ \overline{\sigma_{\ell}(p)} & 1 - \gamma_{\ell}(p) \end{pmatrix} = \frac{1}{1 + e^{H_{\Delta_{\ell}}(p)/T}}.
	\label{eq:EL}
\end{equation}
The right-hand side of Eq.~\eqref{eq:EL} depends only on $\sigma_{\ell}$ through $\Delta_{\ell}$ but not on $\gamma_{\ell}$. That is, $\gamma_{\ell}$ is determined by $\sigma_{\ell}$. 

Let us define $E_{\ell}(p) = \sqrt{(p^2 - \mu)^2 + \vert \Delta_{\ell}(p) \vert^2}$ and the function $\smash{K^{\Delta_{\ell}}_T}$, which for $T > 0$ is given by
\begin{equation*}
K_T^{\Delta_{\ell}}(p) = \frac{E_{\ell}(p)}{\tanh\left(E_{\ell}(p)/(2T) \right)}.
\end{equation*}
Then $\smash{K^{\Delta_{\ell}}_T} = \smash{K^{\Delta_{\ell}}_T}(- i\grad)$ defines an operator on $L^2(\R^2)$ acting by multiplication with $\smash{K^{\Delta_{\ell}}_T}(p)$ in Fourier space. 
Calculations given explicitly in~\cite{HaiSei16} show that~\eqref{eq:EL} is equivalent to 
\begin{align}
\gamma_{\ell}(p) & = \frac{1}{2} - \frac{p^2 - \mu}{2K_{T}^{\Delta_{\ell}}(p)}, \label{eq:ELgamma1} \\
\sigma_{\ell}(p) & = - \frac{\Delta_{\ell}( p)}{2K_{T}^{\Delta_{\ell}}(p)}.  \label{eq:ELalpha}
\end{align}
Using Eq.~\eqref{eq:gapfunction}, we see that Eq.~\eqref{eq:ELalpha} can be written as
\begin{align} \label{eq:gap}
\left( K_T^{\Delta_{\ell}} + V_{\ell} \right) \sigma_{\ell} = 0.
\end{align}
We will also make use of this equation in the form
\begin{align} \label{eq:gapalpha}
\left( K_T^{\Delta_{\ell}} + V \right) \alpha_{\ell} = 0,
\end{align}
where~$\alpha_{\ell}$ is of the form~\eqref{eq:formofalpha}.

\section{Proof of Theorem \ref{thm:vollesFunktional} and Theorem \ref{thm:3dim}}  \label{section:proofofthmformofmin}

We begin with the proof of Theorem~\ref{thm:vollesFunktional}. Let $(\gamma_{\ell_0}, \sigma_{\ell_0}) \in \D_{\ell_0}$ be a minimizer of~$\F_{\ell_0}^{\mathrm{ti}}$ and assume $T_c = T_c(\ell_0)$.
Let~$\Gamma_{\ell_0}$ be the BCS state given by the pair $(\gamma_{\ell_0}, \alpha_{\ell_0})$ with $\hat{\alpha}_{\ell_0}(p) = e^{i \ell_0 \varphi}\sigma_{\ell_0}(p)$. Our aim is to show that the inequality $\F(\Gamma) - \F(\Gamma_{\ell_0}) \geq 0$ holds for all $\Gamma \in \D$. We will use a generalization of the trace per unite volume, which for a periodic operator $A$ on $L^2(\R^2, \C^2)$ is defined by 
\begin{align*}
\Trs[A] = \Tro[P_0AP_0 + Q_0AQ_0]
\end{align*}
with
\begin{align*}
 P_0 = \left( \begin{array}{cc} 1 & 0 \\ 0 & 0 \end{array} \right) \text{ and }  Q_0 = \left( \begin{array}{cc} 0 & 0 \\ 0 & 1 \end{array} \right).
\end{align*}
Note that if~$A$ is locally trace class, then $\Trs[A] = \Tro[A]$.

We begin by calculating the difference $\F(\Gamma) - \F(\Gamma_{\ell})$, where~$\Gamma_{\ell}$ corresponds to a minimizer of~$\F_{\ell}^{\mathrm{ti}}$ as described above. The state $\Gamma$ is defined by the pair $(\gamma, \alpha)$. We find
\begin{align} \label{eq:difference}
\F&(\Gamma) - \F(\Gamma_{\ell})  \\ 
& = \Tro\left[\left( -\grad^2 - \mu\right) \left(\gamma - \gamma_{\ell} \right)\right]  \notag \\
 &  \ \ \ \ + \int_{\Omega \times \R^2} V(x-y) \left( \vert \alpha(x,y) \vert^2 - \vert \alpha_{\ell}(x,y) \vert^2 \right) \td(x,y) - T\left( S(\Gamma) - S(\Gamma_{\ell}) \right).  \notag
\end{align}
First, we complete the square in the difference of the interaction terms, which yields
\begin{align*} 
& \int_{\Omega \times \R^2} V(x - y) \left( \vert \alpha(x,y) \vert^2 - \vert \alpha_{\ell}(x,y) \vert^2 \right) \, \td(x,y) \notag \\
& \  = \int_{\Omega \times \R^2} V(x- y) \left( \vert \alpha(x,y)  -  \alpha_{\ell}(x,y) \vert^2 \right) \, \td(x,y) \notag \\  
& \ \ \ \ \ \ \ - 2\int_{\Omega \times \R^2} V(x-y) \left( \vert \alpha_{\ell}(x,y)\vert^2 - \operatorname{Re}\left( \overline{\alpha(x,y)}\alpha_{\ell}(x,y) \right) \right) \, \td(x,y).
\end{align*}
Next, we combine the second term on the right hand side and the first term on the right hand side of~\eqref{eq:difference}. Let $\tilde{\Delta}_{\ell}(p) = e^{i \ell \varphi} \Delta_{\ell}(p)$ where $\varphi$ denotes the angle of $p \in \mathbb{R}^2$ in polar coordinates and $\Delta_{\ell}$ is given by Eq.~\eqref{eq:gapfunction}. Inserting the equation~$\hat{\alpha}_{\ell}(p)= -\tilde{\Delta}_{\ell}(p)/(2 K_T^{\Delta_{\ell}}(p))$ which follows from Eq.~\eqref{eq:ELalpha}, we see that
\begin{align*}
\Tro & \left[\left( -\grad^2 - \mu\right) \left(\gamma - \gamma_{\ell} \right)\right]  \\
&  + 2 \operatorname{Re} \int_{\Omega \times \R^2} V(x-y) \left( \alpha_{\ell}(x,y)\overline{\alpha(x,y)} - \vert \alpha_{\ell}(x,y) \vert^2 \right) \, \td(x,y) \\ 
& = \frac{1}{2} \Trs \left[H_{\tilde{\Delta}_{\ell}} \left( \Gamma - \Gamma_{\ell}\right)\right].
\end{align*}
Here~$H_{\tilde{\Delta}_{\ell}}$ is given as in Eq.~\eqref{def:HDelta} with $\Delta_{\ell}$ replaced by $\tilde{\Delta}_{\ell}$.

At this point, it turns out to be convenient to introduce the relative entropy~$\H$, which for two BCS states $\Gamma, \tilde{\Gamma} \in \D$ is given by 
\begin{align*}
\mathcal{H}(\Gamma, \tilde{\Gamma}) = \Trs\left[ \Gamma\left(\log\Gamma - \log\tilde{\Gamma}\right) + \left(1-\Gamma\right)\left(\log\left(1-\Gamma\right) - \log\left(1-\tilde{\Gamma}\right)\right)\right].
\end{align*}

The fact that $H_{\tilde{\Delta}_{\ell}}/T = \log(1-\Gamma_{\ell})-\log\Gamma_{\ell}$ yields the following statement.

\begin{lemma} \label{lem:diff}
Let $(\gamma_{\ell}, \sigma_{\ell}) \in \mathcal{D}_{\ell}$ be a minimizer of~$\F_{\ell}^{\mathrm{ti}}$ and let~$\Gamma_{\ell}$ be given by the pair $(\gamma_{\ell}, \alpha_{\ell})$ where $\alpha_{\ell}(p) = e^{i \ell \varphi}\sigma_{\ell}(p)$. 
Then 
\begin{align*}
\F(\Gamma) - \F(\Gamma_{\ell})  & = \frac{T}{2} \mathcal{H}\left(\Gamma, \Gamma_{\ell}\right) + \int_{\Omega \times \R^2} V(x-y) \vert \alpha(x,y) - \alpha_{\ell}(x,y) \vert^2 \, \td(x,y)
\end{align*}
for all $\Gamma \in \mathcal{D}$, where $\alpha = (\Gamma)_{12}$.
\end{lemma}

Based on this identity, we estimate $\F(\Gamma) - \F(\Gamma_{\ell_0})$ from below by applying the relative entropy inequality \cite{FHSS12,HaiLewSei08}.

\begin{proposition} \label{prop:nonnegative}
Let $(\gamma_{\ell},\sigma_{\ell}) \in\mathcal{D}_{\ell}$, be a minimizer of~$\F_{\ell}^{\mathrm{ti}}$, let~$\Gamma_{\ell}$ be as in Lemma~\ref{lem:diff} and denote $V_y(x) = V(x-y)$. Then, for all $\Gamma \in \mathcal{D}$, with $\alpha = (\Gamma)_{12}$, 
\begin{align*}
\F(\Gamma) - \F(\Gamma_{\ell}) \geq & \int_{\Omega} \left\langle \alpha, \left( K_{T}^{\Delta_{\ell}} + V_y(x) \right)_x \alpha \right\rangle_{L^2(\R^2,\td x)} \, \td y  \\
& + \Tro K_T^{\Delta_{\ell}}(\Gamma - \Gamma_{\ell})^2.
\end{align*}
Here, we understand $(K_{T}^{\Delta_{\ell}} + V_y(x))_x$ as an operator acting on the $x$-coordinate of $\alpha(x,y)$.
\end{proposition}

\begin{proof}
The claimed estimate is a consequence of an inequality for the relative entropy that has been proven in~\cite[Lemma 5]{FHSS12}. An application of this inequality yields
\begin{align*}
\F(\Gamma) - \F(\Gamma_{\ell}) & \geq \frac{1}{2} \Tro \left[ \left( \Gamma - \Gamma_{\ell}\right) \frac{H_{\tilde{\Delta}_{\ell}}}{\tanh\left(H_{\tilde{\Delta}_{\ell}}/(2T) \right)}  \left( \Gamma - \Gamma_{\ell}\right)\right] \\
 & \ \ \ + \int_{\Omega \times \R^2} V(x-y) \vert \alpha(x,y) - \alpha_{\ell}(x,y) \vert^2 \, \td(x,y). 
\end{align*}
The fact that $x \mapsto x(\tanh(x/2))^{-1}$ is an even function and 
\begin{align*}
H_{\tilde{\Delta}_{\ell}}^2(p) = \mathbb{I}_{\C^2} E^2_{\ell}(p)
\end{align*}
is diagonal, implies the statement.
\end{proof}

Next, we show that the operator $\smash{K^{\Delta_{\ell_0}}_T} + V$ is nonnegative for $T \in [\tilde{T}, T_c)$. 

\begin{proposition}\label{prop:nonnegativity}
Assume $V \in L^2(\mathbb{R}^2)$ and $\hat{V} \in L^r(\R^2)$ for some $r \in [1,2)$. If the lowest eigenvalue of $K_{T_c} + V$ is at most twice degenerate then there exists $\tilde{T} < T_c$ such that $\smash{K^{\Delta_{\ell_0}}_T} + V$ is nonnegative as an operator on $L^2(\R^2)$ for all $T \in [\tilde{T}, T_c)$.
\end{proposition}

The proof of Proposition \ref{prop:nonnegativity} is based on spectral perturbation theory and relies on the fact that $\smash{K^{\Delta_{\ell_0}}_{T}} + V \to K_{T_c} + V$, while $\Delta_{\ell_0}(T) \to 0$, in norm resolvent sense for $T \to T_c$. We will derive this convergence from the following lemmas. 
In order to simplify the notation we write $a \lesssim b$ if there exists a constant $c > 0$ such that $a \leq cb$. Moreover, we denote by $\| \cdot \|$ the operator norm and by $\|\cdot \|_{r}$ the $L^r(\R^2)$-norm.

\begin{lemma} \label{lem:AB1}
Let $T \in (0,T_c)$. The operators $K_{T_c} - K_{T}$ and $\smash{K^{\Delta_{\ell_0}}_{T}} - K_{T}$ are bounded. More precisely, $\| K_{T_c} - K_{T} \| \lesssim (T_c - T)$ and $\| \smash{K^{\Delta_{\ell_0}}_{T}} - K_{T} \| \lesssim \|\Delta_{\ell_0} \|_{\infty}$. Moreover, $K_{T_c} - K_{T} \geq 0$ and $\smash{K^{\Delta_{\ell_0}}_{T}} - K_{T} \geq 0$.
\end{lemma}

\begin{proof}
In the proof we abbreviate $A_{T} \coloneqq  K_{T_c} - K_{T}$ and $B_{T} \coloneqq \smash{K^{\Delta_{\ell_0}}_{T}} - K_{T}$.
Notice that
\begin{align*}
K^{\Delta_{\ell_0}}_T(p) = \frac{\sqrt{k(p)^2 + \vert \Delta_{\ell_0}(p) \vert^2}}{\tanh\left( \sqrt{k(p)^2 + \vert \Delta_{\ell_0}(p) \vert^2}/(2T) \right)}
\end{align*}
is an increasing function in~$T$ for fixed~$\Delta_{\ell_0}$ and vice versa. Hence $A_{T} \geq 0$ and $B_{T} \geq 0$. Both, $A_T$ and~$B_T$ are pseudo-differential operators and by a slight abuse of notation we denote by~$A_T(p)$ and~$B_T(p)$ the symbols of~$A_T$ and~$B_T$, respectively.
In the following we abbreviate $T_c - T = \delta T$ and
\begin{align*}
I_{T} = \frac{1}{T} - \frac{1}{T_c}. 
\end{align*}
A simple calculation yields
\begin{align*}
A_T(p) = \int_0^1 \frac{I_{T} k(p)^2}{2\sinh^2\left(k(p)/(2T_c) + t I_{T} k(p)/2 \right)} \, \td t.
\end{align*}
Obviously, for large~$\vert p \vert$ the smooth function $A: p \mapsto A(p)$ and all its derivatives have exponential decay. Moreover, $\vert I_{T} \vert \lesssim T_c-T$ implies $\| A_T \| \lesssim T_c - T$. 
In order to derive an analogous representation for~$B_T(p)$ we define
\begin{align} \label{eq:der1}
f(x) \coloneqq \frac{\td}{\td x} \frac{x}{\tanh(x/(2T))} = \frac{T\sinh(x/T) - x}{2T \sinh^2(x/(2T))}  
\end{align} 
as well as
\begin{align} \label{def:deltaE}
\delta E_{\ell_0}(p) = 
\sqrt{k(p)^2 + \vert \Delta_{\ell_0}(p) \vert^2} - \vert k(p) \vert.
\end{align}
A straightforward calculation shows that
\begin{align} \label{eq:B}
B_T(p) =  \delta E_{\ell_0}(p) \, \int_0^1 f(\vert k(p) \vert + t \delta E_{\ell_0}(p))  \, \td t.
\end{align}
Since the function~$f$ defined in~\eqref{eq:der1} is bounded by~$1$, we find that $\vert B_T(p) \vert \leq  \vert \delta E_{\ell_0}(p) \vert$ for all $p \in \R^2$. It can be seen directly from the definition of $\delta E_{\ell_0}(p)$, see~\eqref{def:deltaE}, that $\vert \delta E_{\ell_0}(p) \vert \leq  \vert \Delta_{\ell_0}(p)\vert$ for all $p \in \R^2$, which implies $\vert B_T(p) \vert \leq \vert \Delta_{\ell_0}(p) \vert$ for all $p \in \R^2$. 
\end{proof}

\begin{lemma} \label{lem:AB1/2}
Let $T \in (0,T_c)$. If $\alpha_{\ell_0}$ is a solution of the BCS gap equation in the form of Eq.~\eqref{eq:gapalpha}, then $\|(1 + p^2)^{1/4} \hat{\alpha}_{\ell_0}\|_4^4 \lesssim \langle \alpha_{\ell_0}, (\smash{K^{\Delta_{\ell_0}}_{T}} - K_{T})  \alpha_{\ell_0} \rangle$.
\end{lemma}

\begin{proof}
We will make use of the following observation, which is implied by the fact that the function $\vert \Delta_{\ell_0} \vert \mapsto \vert \Delta_{\ell_0} \vert/\smash{K^{\Delta_{\ell_0}}_{T}}$ is strictly increasing. Eq.~\eqref{eq:gapfunction} implies that 
\begin{align} \label{est:DeltaValpha}
\| \Delta_{\ell_0} \|_{\infty} \leq \| V \|_2 \| \hat{\alpha}_{\ell_0} \|_2.
\end{align}
We will abbreviate $\| V \|_2 \| \hat{\alpha}_{\ell_0} \|_2$ by $c(\alpha_{\ell_0})$ in the following.
Thus, together with \eqref{eq:ELalpha}, the just mentioned monotonicity of $\vert \Delta_{\ell_0} \vert/\smash{K^{\Delta_{\ell_0}}_{T}}$ implies that
\begin{align*}
\vert \hat{\alpha}_{\ell_0}(p) \vert \leq \frac{c(\alpha_{\ell_0})}{2 K^{c(\alpha_{\ell_0})}_{T}(p)}
\end{align*}
for all $p \in \R^2$. By taking the square and integrating, we see that
\begin{align*}
1 & \leq \frac{\| V \|_2^2}{4} \int_{\R^2} \left(K^{c(\alpha_{\ell_0})}_{T}(p)\right)^{-2} \, \td p.
\end{align*}
Next, we use that $\tanh(x) \leq 1$ for all $x$, which leads to
\begin{align*}
1 & \leq \frac{\| V \|_2^2}{4} \int_{\R^2} \left((p^2 - \mu)^2 +\| V \|_2^2 \| \hat{\alpha}_{\ell_0} \|_2^2\right)^{-1} \, \td p
\end{align*}
We may assume that $\| V \|_2^2 \| \alpha_{\ell_0} \|_2^2 \geq \mu^2$ and conclude that
\begin{align*}
1 \leq \frac{\| V \|_2^2}{4} \int_{\R^2} \left( p^4/2  - \mu^2 + \| V \|_2^2 \| \hat{\alpha}_{\ell_0} \|_2^2\right)^{-1} \, \td p.
\end{align*}
From this estimate one easily derives that
\begin{align*} 
\| \hat{\alpha}_{\ell_0} \|_2^2 \leq \frac{\| V \|_2^2 \pi^4}{32} + \frac{\mu^2}{\| V \|_2^2}.
\end{align*} 
Making use of \eqref{est:DeltaValpha}, we see that this directly implies that
\begin{align} \label{est:deltauni}
\| \Delta_{\ell_0} \|_{\infty}^2 \leq \frac{\| V \|_2^4 \pi^4}{32} + \mu^2.
\end{align} 
In other words, there exists a constant $m > 0$ that only depends on $V$ and $\mu$, such that $\vert \Delta_{\ell_0}(p) \vert < m$ for all $p \in \R^2$. In particular, $m$ does not depend on $T$.

We have to estimate $\smash{K_T^{\Delta_{\ell_0}}} - K_T$ from below. We recall that $\vert \Delta_{\ell_0} \vert \mapsto \smash{K_T^{\Delta_{\ell_0}}}/\vert \Delta_{\ell_0} \vert^2$ is decreasing. Having in mind that $\smash{K_T^{\Delta}}-K_T$ behaves like $\vert \Delta \vert^2$ for small $\vert \Delta \vert$ we thus estimate 
\begin{align*}
\frac{K_T^{\Delta_{\ell_0}} - K_T}{\vert \Delta_{\ell_0} \vert^2}\vert \Delta_{\ell_0} \vert^2 \gtrsim  \left(\frac{K_T^m - K_T}{m^2} \right) \vert \Delta_{\ell_0} \vert^2 .
\end{align*}
Abbreviating $y_t = \sqrt{k(p)^2 + tm^2}/(2T)$ we find that
\begin{align} \label{est:below1}
K_T^{\Delta_{\ell_0}}(p) - K_T(p) & = 2T \int_0^1 \frac{\td}{\td t} \frac{y_t}{\tanh\left(y_t\right)} \, \td t \notag \\ 
& = \frac{m^2}{4T} \int_0^1 \left( \frac{1}{y_t \tanh(y_t)} - \frac{1}{\sinh^2(y_t)} \right) \td t.
\end{align}
As one easily sees, the function 
\begin{align*}
g(y) = \frac{1}{y}\frac{1}{\tanh(y)} - \frac{1}{\sinh^2(y)}
\end{align*}
is decreasing, which implies
\begin{align*}
 K_T^{\Delta_{\ell_0}}(p) - K_T(p) \gtrsim \frac{m^2}{4T}\left( \frac{1}{y_1 \tanh(y_1)} - \frac{1}{\sinh^2(y_1)} \right).
\end{align*}
Moreover, $g$ is bounded from below by $g(y) \geq 2/3 \,(1 + y)^{-1}$.
Together with \eqref{est:below1} this shows that
\begin{align} \label{est:below2}
K_T^{\Delta_{\ell_0}}(p) - K_T(p) \gtrsim \vert \Delta_{\ell_0}(p) \vert^2 \frac{1}{1+p^2}.
\end{align}
Next, we make use of the Euler-Lagrange equation of $\F_{\ell_0}^{\mathrm{ti}}$, that is the relation $\vert\Delta_{\ell_0}(p)\vert = 2\smash{K_T^{\Delta_{\ell_0}}}(p)\vert \hat{\alpha}_{\ell_0}(p)\vert$. Inserting this identity in \eqref{est:below2} we see that
\begin{align*}
K_T^{\Delta_{\ell_0}}(p) - K_T(p) \gtrsim \left(K_T^{\Delta_{\ell_0}}(p) \right)^2 \frac{\vert \hat{\alpha}_{\ell_0}(p) \vert^2}{1 + p^2} \gtrsim \left( 1 + p^2 \right) \vert \hat{\alpha}_{\ell_0}(p) \vert^2,
\end{align*}
which implies the statement.
\end{proof}

\begin{lemma} \label{lem:AB2} 
Let $T \in (0,T_c)$. If $\alpha_{\ell_0}$ is a solution of the BCS gap equation in the form~\eqref{eq:gapalpha}, then $\| \alpha_{\ell_0}\|_2 \lesssim (T_c - T)^{1/2}$. In particular, $\|\Delta_{\ell_0} \|_{\infty} \lesssim  (T_c - T)^{1/2}$.
\end{lemma}

\begin{proof}
The gap equation, see~\eqref{eq:gapalpha}, can be written as
\begin{align*}
\langle \alpha_{\ell_0}, \left( K_{T_c} + V \right) \alpha_{\ell_0}\rangle + \langle \alpha_{\ell_0}, B\alpha_{\ell_0} \rangle = \langle \alpha_{\ell_0}, A \alpha_{\ell_0}\rangle,
\end{align*}
where we use the notation introduced in the proof of Lemma~\ref{lem:AB1} but drop the subscript, i.e. $A = A_T$ and $B = B_T$ for brevity.
Lemma \ref{lem:AB1} and the definition of $T_c$ imply that 
\begin{align} \label{est:B}
 \langle \alpha_{\ell_0}, B \alpha_{\ell_0} \rangle \leq  \langle \alpha_{\ell_0}, A \alpha_{\ell_0}\rangle \lesssim  (T_c - T) \|\alpha_{\ell_0}\|_2^2. 
\end{align}
From the combination of Lemma~\ref{lem:AB1/2} and ~\eqref{est:B} we deduce that
\begin{align*}
\| \left(1 + p^2\right)^{1/4} \hat{\alpha}_{\ell_0} \|_4^4 \lesssim   (T_c - T)  \| \alpha_{\ell_0} \|_2^2.
\end{align*}
On the other hand, the $L^r(\R^2)$-norm of~$\hat{\alpha}$ is bounded from above by
\begin{align*}
\| \hat{\alpha}_{\ell_0} \|_r \leq \| \left( 1 + p^2\right)^{-1/4} \|_s \| \left( 1 + p^2\right)^{1/4} \hat{\alpha}_{\ell_0} \|_4,
\end{align*}
where $r > 2$, due to the fact that we have to choose $s > 4$. Thus, 
\begin{align} \label{eq:alpha42}
\| \hat{\alpha}_{\ell_0} \|_r^4 \lesssim   (T_c - T) \| \hat{\alpha}_{\ell_0} \|_2^2. 
\end{align}
Furthermore, we conclude from the relation between~$\Delta_{\ell_0}$ and~$\alpha_{\ell_0}$ given by Eq.~\eqref{eq:gapfunction} that
\begin{align} \label{eq:DeltaAlpha}
\| \Delta_{\ell_0} \|_{\infty} \lesssim \| \hat{V} \|_t \| \hat{\alpha}_{\ell_0} \|_r,
\end{align}
where we choose $r > 2$ and $t \in [1,2)$ appropriately. Note that the gap equation in the form~\eqref{eq:ELalpha} implies that $\| \hat{\alpha}_{\ell_0} \|_2 \lesssim \| \Delta_{\ell_0} \|_{\infty}$. Together with~\eqref{eq:alpha42} and~\eqref{eq:DeltaAlpha} this finally shows that
\begin{align*}
\| \hat{\alpha}_{\ell_0} \|_2 \lesssim (T_c - T)^{1/4} \| \hat{\alpha}_{\ell_0} \|_2^{1/2}
\end{align*}
and hence proves the first part of the claim. In order to get the estimate on $\| \Delta_{\ell_0} \|_{\infty}$, we go back to \eqref{eq:alpha42} and insert $\| \alpha_{\ell_0} \|_2 \lesssim (T_c - T)^{1/2}$. Together with \eqref{eq:DeltaAlpha} this yields the statement. 
\end{proof}

Let $T \in (0, T_c)$ and $z \in \C\setminus\R$. 
Taken together, Lemma \ref{lem:AB1} and Lemma \ref{lem:AB2} show that
\begin{align*}
& \left\|\left( z - \left( K_{T_c} + V \right)\right)^{-1} - \left( z - \left( K^{\Delta_{\ell_0}}_{T} + V\right)\right)^{-1}\right\| \\
& \ \ \ \leq \left\|\left( z - \left( K_{T_c} + V \right)\right)^{-1} \right\| \left\|  K^{\Delta_{\ell_0}}_{T} - K_{T_c} \right\| \left\| \left(z - \left( K^{\Delta_{\ell_0}}_{T} + V\right)\right)^{-1} \right\| \\
& \ \ \ \lesssim \vert \operatorname{Im}(z) \vert^{-2}  (T_c - T)^{1/2}.
\end{align*}
In other words, $\smash{K^{\Delta_{\ell_0}}_{T}} + V \to K_{T_c} + V$ for $T \to T_c$ in norm resolvent sense for an arbitrary $z \in \C\setminus \R$ and consequently for all $z \in \rho(K_{T_c} + V)$.

We are now prepared for the proof of Proposition~\ref{prop:nonnegativity}.

\begin{proof}[Proof of Proposition \ref{prop:nonnegativity}]
We consider the case $\ell_0 \neq 0$. The proof for the case $\ell_0 = 0$ is analogous.
As illustrated in Figure 1, we have by assumption that $T_c = T_c(\ell_0)$ and that the lowest eigenvalue of $K_{T_c} + V$ is exactly twice degenerate. Note that in the case that $\ell_0 = 0$ the smallest eigenvalue is non-degenerate. From the convergence of $\smash{K^{\vphantom{b}\Delta_{\ell_0}}_{T} + V}$ to $K_{T_c} + V$ in norm resolvent sense one concludes that the lowest eigenvalue of $\smash{K^{\vphantom{b}\Delta_{\ell_0}}_{T} + V}$ is stable.

In particular, this tells us that there exists~$\tilde{T} < T_c$ such that $\smash{K^{\Delta_{\ell_0}}_{T} + V}$ with $T \in (\tilde{T}, T_c]$ has exactly two eigenvalues $\lambda_1(T), \lambda_2(T)  \in  \{ z \in \C \vert \ \vert z \vert < r \}$ for some radius $r > 0$. 
Combining this with the fact that the Euler-Lagrange equation \eqref{eq:gapalpha} of~$\F_{\ell_0}^{\mathrm{ti}}$ reads 
\begin{align} \label{eq:ELKT}
(K^{\Delta_{\ell_0}}_T + V) \alpha = 0,
\end{align}
we conclude that 
$\lambda_1(T) =  \lambda_2(T) = 0$. Having in mind that $\smash{K^{\Delta_{\ell_0}}_T}$ is an  increasing function of $T$ and of $\Delta_{\ell_0}$, what we have seen by this argument is that the effects of these monotonicity properties exactly correspond.
In other words, we have shown that there exists $\tilde{T} < T_c$ such that $\smash{K^{\Delta_{\ell_0}}_T} + V$ is nonnegative for all $T \in [\tilde{T}, T_c]$. It is not hard to see that $\tilde{T}$ can be chosen as pointed out in Remark \ref{rem:ell1}.
\end{proof}

\begin{proof}[Proof of Theorem \ref{thm:vollesFunktional}] We know from Lemma~\ref{lem:extmin} that for~$\ell_0$ determined by $T_c(\ell_0) = \max_{\ell \in 2\N} T_c(\ell)$ the functional~$\F_{\ell_0}^{\mathrm{ti}}$ has a minimizer $(\gamma_{\ell_0}, \sigma_{\ell_0})$. Proposition~\ref{prop:nonnegative} and Proposition~\ref{prop:nonnegativity} show that for~$\Gamma_{\ell_0}$ given by $(\gamma_{\ell_0}, \alpha_{\ell_0})$, with $\alpha_{\ell_0}$ as in \eqref{eq:formofalpha},
\begin{align*}
\F(\Gamma) - \F(\Gamma_{\ell_0}) \geq 0,
\end{align*}
holds for all $\Gamma \in \mathcal{D}$.
Moreover, if $\F(\Gamma) - \F(\Gamma_{\ell_0}) = 0$, then $\gamma = \gamma_{\ell_0}$ and $\alpha \in \ker(K_T^{\Delta_{\ell_0}} + V_y)$ by Proposition~\ref{prop:nonnegative}. Consequently, $\alpha$ takes the form $\alpha = \psi_1\alpha_{\ell_0} + \psi_2\alpha_{-\ell_0}$, where $\alpha_{\pm \ell_0}(p) = e^{\pm i \ell\varphi} \sigma_{\ell_0}(p)$ and $\psi_1$ and $\psi_2$ denote complex constants. It remains to show that either $\psi_1 = 0$ and $\vert \psi_2 \vert = 1$ or $\vert \psi_1 \vert = 1$ and $\psi_2 = 0$. Observe that, in particular, $(\gamma_{\ell_0}, \alpha) \in \D^{\mathrm{ti}}$ and as we know that~$\F^{\mathrm{ti}}$ has a minimizer, we conclude that $(\gamma_{\ell_0}, \alpha)$ satisfies the Euler-Lagrange equation of~$\F^{\mathrm{ti}}$, that is
\begin{align*}
\gamma_{\ell_0}(p) = \frac{1}{2} - \frac{p^2 - \mu}{2K_T^{\Delta}(p)},
\end{align*}
where $\Delta = \pi^{-1} \hat{V}\ast \hat{\alpha}$. Hence $\vert\Delta\vert$ is a radial function and consequently either $\psi_1= 0$ or $\psi_2 = 0$. In other words, $(\gamma_{\ell_0}, \sigma_{\ell_0}) \in \D_{\ell_0}$. Thus, in order to find minimizers of~$\F$, it is sufficient to find the minimizers of~$\F_{\ell_0}^{\mathrm{ti}}$. As we know that~$\F_{\ell_0}^{\mathrm{ti}}$ has minimizers, the only thing left to show is that $(\gamma_{\ell_0}, \sigma_{\ell_0})$ is, up to a phase, the only minimizer of~$\F_{\ell_0}^{\mathrm{ti}}$. The fact that other possible minimizers $(\gamma_{\ell_0}, \psi\sigma_{\ell_0})$, for some $\psi \in \C$, have to satisfy the gap equation~\eqref{eq:gap} of~$\F_{\ell_0}^{\mathrm{ti}}$ reads 
\begin{align*}
\left(K_T^{\psi \Delta_{\ell_0}} + V_{\ell_0} \right) \left( \psi \sigma_{\ell_0}\right) = 0.
\end{align*}
Together with the monotonicity of $K_T^{\psi \Delta_{\ell_0}}$ in~$\psi$ this implies that $\vert \psi \vert$ = 1. 
\end{proof}

The proof of Theorem~\ref{thm:3dim} is analogous to the proof of Theorem~\ref{thm:vollesFunktional} with one exception. 

\begin{proof}[Proof of Theorem~\ref{thm:3dim}]
In case $\ell_0 = 0$ all given arguments also apply in the three-dimensional case. The only exception is Lemma~\ref{lem:AB2}, where we need to modify the assumptions slightly. One easily sees that $\hat{V} \in L^r(\R^3)$ with $r \in [1, 12/7)$ is a sufficient assumption in this case. 
\end{proof}

\begin{proof}[Proof of Proposition~\ref{prop:mainproposition}]
	We will carry out the proof for $d=3$ and afterwards comment on the case $d=2$. The Cooper-pair wave function of any minimizer of the translation-invariant BCS functional satisfies $\hat{\alpha}(p) = -\Delta(p)/(2K_T^{\Delta}(p))$ which is implied by the Euler-Lagrange equation of~$\F$, see \cite{HHSS08} or compare with Section~\ref{sec:prep}. Hence, $\vert \hat{\alpha}\vert$ is radial if and only if $\vert \Delta \vert$ is radial. With Eq.~\eqref{eq:proposition1} and the assumption that $V$ is a radial function, one checks that it is sufficient to show
	\begin{equation}
	\left\langle U(R) \alpha, K_T^{\Delta}U(R)\alpha  \right\rangle < \left\langle \alpha, K_T^{\Delta}\alpha  \right\rangle.
	\label{eq:rearrangement2}
	\end{equation}
	Using the above relation between $\hat{\alpha}$ and $\Delta$, we write
	\begin{align}
	\left\langle U(R)\alpha, K_T^{\Delta}U(R)\alpha  \right\rangle & = \frac{1}{4} \int_{\R^3} \frac{\vert\Delta(p)\vert^2}{K_T^{\Delta}(p)^2}K_T^{\Delta}(Rp)\, \td p \nonumber \\
	& = \frac{1}{4} \int_0^{\infty} \int_{\Omega_r} \frac{\vert\Delta(p)\vert^2}{K_T^{\Delta}(p)^2}K_T^{\Delta}(Rp)\, \td \omega(p) \, r^2 \td r, \nonumber
	\end{align}
	where $\Omega_r$ denotes the sphere with radius $r$ and $\td \omega(p)$ denotes the uniform measure on $\Omega_r$. On $\Omega_r$, that is, for fixed radius $r = \vert p\vert$, we can understand $\vert\Delta(p)\vert^2/K_T^{\Delta}(p)^2$ as a function $f$ that depends only on $\vert \Delta(p) \vert$. There also exists a function $g$ such that $K_T^{\Delta}(Rp) = g(\vert \Delta(Rp) \vert)$ for all $p \in \Omega_r$. The functions $f$ and $g$ are both strictly increasing. 
	
	Consider the expression
	\begin{equation}
	M(R) := \int_{\Omega_r} [g(\Delta(Rp)) - g(\Delta(p))][f(\Delta(Rp)) - f(\Delta(p))] d \omega(p) \nonumber
	\end{equation}
	The functions $f$ and $g$ depend only on the magnitude of $\Delta(Rp)$ resp. $\Delta(p)$.
	Since $f$ and $g$ are strictly increasing we have that $M(R)>0$ unless $|\Delta(Rp)| = |\Delta(p)|$ for a.e. $p$.
	To see this assume that $|\Delta(Rp)|$ and $|\Delta(p)|$ differ on a set of positive measure.
	Now consider the set $\{p: |\Delta(Rp)| > |\Delta(p)|\}$ and the set $\{p: |\Delta(Rp)| < |\Delta(p)|\}$
	At least one of them  must have positive measure. Hence on the union of these sets
	\begin{equation}
	[g(\Delta(Rp)) - g(\Delta(p))][f(\Delta(Rp)) - f(\Delta(p))] >0 \nonumber
	\end{equation}
	since $f$ and $g$ are both strictly increasing.
	Using the rotation invariance of the measure $\omega$, we find
	\begin{align}
	0< M(R)= 2\int_{\Omega_r} g(\Delta(p))f(\Delta(p)) d\omega(p)  &- \int_{\Omega_r} g(\Delta(p))f(\Delta(Rp)) d \omega(p)  \nonumber \\
	&-\int_{\Omega_r} g(\Delta(Rp) ) f(\Delta(p))] d \omega(p) \nonumber
	\end{align}
	and hence one of the integrals
	\begin{equation}
	\int_{\Omega_r} g(\Delta(p))f(\Delta(Rp)) d \omega(p) \nonumber
	\end{equation}
	or 
	\begin{equation}
	\int_{\Omega_r} g(\Delta(Rp) ) f(\Delta(p)) d \omega(p) \nonumber
	\end{equation}
	must be strictly below
	\begin{equation}
	\int_{\Omega_r} g(\Delta(p))f(\Delta(p)) d\omega(p). \nonumber
	\end{equation}
	Accordingly, there exists a $R \in SO(3)$ such that
	\begin{equation}
	\int_{\Omega_r} \frac{\vert\Delta(p)\vert^2}{K_T^{\Delta}(p)^2}K_T^{\Delta}(Rp)\, \td \omega(p) < \int_{\Omega_r} \frac{\vert\Delta(p)\vert^2}{K_T^{\Delta}(p)^2}K_T^{\Delta}(p)\, \td \omega(p). \label{eq:proposition2}
	\end{equation}
	To conclude that Eq.~\eqref{eq:rearrangement2} holds, it suffices to note that $\Delta$ is a continuous function, see the first paragraph in the proof of \cite[Proposition~3]{HHSS08}, which implies that both sides of Eq.~\eqref{eq:proposition2} are continuous functions of $r$. If $d=2$ the proof goes through in the same way with the only difference that the continuity of $\Delta$ is concluded from $\Delta(p) = \pi^{-1} \hat{V}\ast\hat{\alpha}(p)$, the assumption that $V \in L^2(\mathbb{R}^2)$ and the Riemann-Lebesgue Lemma.
\end{proof}

\subsection*{Acknowledgments}
The paper was partially supported by the GRK 1838 and the Humboldt foundation.
M.L. was partially supported by NSF grant DMS-1600560. Partial financial support by the European Research Council (ERC) under the European Union's Horizon 2020 research and innovation programme (grant agreement No 694227) is gratefully acknowledged (A.D.). We are grateful for the hospitality at the Department of Mathematics at the University of T\"{u}bingen (M.L.) and at the Georgia Tech School of Mathematics (A.G.).

\end{document}